\documentclass{elsarticle} 

\usepackage{amsmath,amssymb}
\usepackage{amsthm}
\usepackage{stmaryrd}
 \usepackage[all,arc,line]{xy}

\newtheorem{theorem}{Theorem}
\newtheorem{lemma}[theorem]{Lemma}
\newtheorem{proposition}[theorem]{Proposition}
\newtheorem{corollary}[theorem]{Corollary}

\theoremstyle{definition}
\newtheorem{definition}{Definition}
\newtheorem{remark}{Remark}

\DeclareMathOperator{\downarr}{\downarrow}
\DeclareMathOperator{\uparr}{\uparrow}

\begin{document}
\title{Canonical extension and canonicity via DCPO presentations}

\author{Mai Gehrke}
\ead{M.Gehrke@math.ru.nl}
\address{Institute for Mathematics, Astrophysics and Particle Physics,
Radboud Universiteit Nijmegen\\ Postbus 9010, 6500 GL Nijmegen, The Netherlands}

\author{Jacob Vosmaer\corref{cor}\fnref{NWO-Vici}}
\ead{J.Vosmaer@uva.nl}
\address{Institute for Logic, Language and Computation, Universiteit van Amsterdam\\ Postbus 94242, 1090 GE Amsterdam, The Netherlands}

\cortext[cor]{Corresponding author.}
\fntext[NWO-Vici]{The research of this author has been made possible by VICI grant
    639.073.501 of the Netherlands Organization for Scientific Research
    (NWO).}

\date{\today}

\begin{keyword}
dcpo presentation, dcpo algebra, lattice theory, canonical extension, canonicity

\end{keyword}

\begin{abstract}
  The canonical extension of a lattice is in an essential way a
  two-sided completion. Domain theory, on the contrary, is primarily
  concerned with one-sided completeness. In this paper, we show two
  things. Firstly, that the canonical extension of a lattice can be
  given an asymmetric description in two stages: a free co-directed
  meet completion, followed by a completion by \emph{selected}
  directed joins. Secondly, we show that the general techniques for
  dcpo presentations of dcpo algebras used in the second stage of the
  construction immediately give us the well-known canonicity result
  for bounded lattices with operators.
\end{abstract}

\maketitle

\section{Introduction}
Domain theory on the one side and canonical extensions and canonicity
on the other side are topics that have played a fundamental role in
non-classical logic and its computer science applications for a long
time. Domain theory has been intrinsically tied to foundational issues
in computer science since it was introduced by Dana Scott in the late
1960s in order to provide semantics for the lambda calculus
\cite{Scott1970}. The solution of domain equations and the modern
techniques for dcpo presentations are particularly important tools
\cite{AJ1994,JMV2008}.  Canonical extensions in their algebraic form
were first introduced by J\'onsson and Tarski in 1951 with the hopes
of giving a representation theorem for relation algebras
\cite{JT1951}. However, they were later realised to be closely related
to the very important canonical model construction in logic and thus
to issues concerning relational semantics for a plethora of logics
important in computer science applications such as modal logics
\cite{Goldblatt1989}.  The algebraic approach to canonical extensions
and questions of canonicity have been revitalised over the last few
decades after the theory was extended beyond the setting of
Boolean-based logics and additional operations that preserve joins in
each coordinate. The initial step in this development was the
realisation that Scott continuity plays a central role in the theory
\cite{GJ1994}.  Apart from this one fundamental connection, the two
topics have not had much to do with each other and any more tangible
connections have remained hidden. This is somewhat remarkable in light
of the central role Stone duality plays in both domain theory \cite{Abr1991}
\emph{and} canonical extension. We will briefly touch upon the
interaction between Stone duality, domain theory and canonical
extenstions in
Section \ref{intro:canext-and-logic} below.

On a more directly mathematical level, there are also other reasons to
seek
to understand the connections between domain theory and the theory of
canonical extensions. Completing, or directedly completing, posets may
be done freely if we only consider one-sided limits in the form either
of joins or meets and this is fundamental to the theory of domains and
the related theory of frames as studied in pointfree
topology. However, unrestricted two-sided free completions do not
exist. Canonical extensions may be viewed as the second level (after
MacNeille completion) of two-sided completions obtained by restricting
the alternations of joins and meets required to generate the
completion \cite{GP2008}.  As such, they are certainly dcpos, and in
the distributive setting, algebraic domains and they remain so when
turned upside down.  This begs the question of understanding these
two-sided completions relative to the one-sided completion techniques
that are so central in domain theory.  In this spirit, this paper is
an answer to a question raised by Achim Jung during his talk at
TACL2009 of the relation between his results with Moshier and Vickers
in \cite{JMV2008} and canonical extensions. To be specific, we show
that the canonical extension of a lattice can be given an asymmetric
description in two stages: a free co-directed meet completion followed
by a completion by selected directed joins as made possible by the
methods of dcpo presentations.  In addition, we show that the pivotal
1994 canonicity result \cite{GJ1994} that introduced Scott continuity
into the theory of canonical extensions may in fact be seen as a
special case of the theorem on representations of dcpo algebras given
in \cite{JMV2008} thus making the connection between the two fields
quite explicit. 
In obtaining the 1994 canonicity result from the one-sided theory, the
setting of dcpo algebras rather than just suplattice algebras is crucial
as the former is needed in order to have a result on the lifting of 
operations available (see Remark~\ref{rem:suplatvsdcpo} in 
Section~\ref{section:results} below).

The organization of this paper is as follows: first, we provide brief
discussions about the background of canonical extension, both in
relation to Stone duality and in relation to logic. After that, in
Section \ref{section:preliminaries}, we provide preliminaries on dcpo
presentations, dcpo algebras, free directed completions and canonical
extensions. The main results are presented in Section
\ref{section:results}; after which we conclude the article with a discussion in
Section \ref{section:discussion}.

The authors would like to thank the referees for their thorough reading
of our manuscript and their thoughtful comments which we are sure have 
made our paper easier to read.

\subsection{Canonical extension and Stone duality}
At its base, canonical extension is an algebraic way of talking 
about Stone's duality for bounded distributive lattices. To see
this consider the following square of functors for which both the inner 
and the outer square commute
\begin{displaymath}
\xy 
(-15,8)*+{\mbox{ DL}}="t0"; (-15,-8)*+{\mbox{ DL}^+}="b0"; 
(15,8)*+{\mbox{ Stone}}="t1"; (15,-8)*+{\mbox{Pos}}="b1"; 
{\ar@<.5ex>^{S} "t0"; "t1"}; 
{\ar@<.5ex>^{{CO}} "t1"; "t0"}; 
{\ar@<.5ex>^{ J^\infty} "b0"; "b1"}; 
{\ar@<.5ex>^{\mathcal U} "b1"; "b0"}; 
{\ar@<-.5ex>_{\sigma} "t0"; "b0"}; 
{\ar@<-.5ex>_{} "b0"; "t0"}; 
{\ar@<.5ex>^{} "t1"; "b1"}; 
{\ar@<.5ex>^{\beta} "b1"; "t1"}; 
\endxy
\end{displaymath} 
Here the upper pair of functors gives the Stone duality for bounded
distributive lattices and spectral spaces, and the lower pair of functors 
gives the
`discrete' duality between completely distributive algebraic lattices
(or complete lattices join-generated by their completely join-prime
elements) and partially ordered sets. This second duality generalises
the very well-known duality between complete and atomic Boolean
algebras and Sets. On objects, it sends a completely distributive
algebraic lattice (DL$^+$) to its poset of completely join-irreducible
elements and a poset to its lattice of upsets.

In the vertical direction, we have natural forgetful functors: DL$^+$s
are in particular DLs, and topological spaces give rise to posets 
via the specialisation order: $x\leq y$ if and only if every open 
containing $x$ also contains $y$. These forgetful functors go in
opposite directions so they are obviously not translations of 
each other across the dualities. Instead, they translate to left
adjoints of each other across the dualities. This brings us to the 
canonical extension. The forgetful functor ${\rm DL}^+\to{\rm DL}$ 
that embeds DL$^+$ as a non-full subcategory of DL has a left 
adjoint $\sigma:~{\rm DL}\to{\rm DL}^+$ and its dual incarnation 
is the forgetful functor from Stone spaces to posets. Moreover,
this left adjoint $\sigma:~{\rm DL}\to{\rm DL}^+$ is a reflector. Thus 
we have, for each DL, an embedding $A\hookrightarrow A^\sigma$; 
this embedding is the \emph{canonical extension}. The dual 
incarnation of the inclusion from ${\rm DL}^+$ to ${\rm DL}$ is 
the left adjoint of the forgetful functor from the category of Stone 
(=spectral) spaces to the category of posets. In the distributive
lattice setting this left adjoint was first identified by Banaschewski in 
\cite{Ba56} and in the Boolean setting it is the very well-known 
Stone-{\v C}ech compactification.

We reiterate that both of the inclusions, ${\rm DL}^+\to{\rm DL}$
and the one of spectral spaces in posets, are inclusions as 
{\it non-full subcategories}: a DL$^+$ morphism is not just a bounded
lattice homomorphism but a complete lattice homomorphism; 
similarly there are maps between spectral spaces which preserve 
the specialisation order without being continuous. As a consequence,
even for objects in the subcategories on either side of the square, the 
reflectors need not be the identity. To whit, for an infinite powerset 
Boolean algebra, $B$, the canonical extension will be the powerset of 
the set of all ultrafilters of $B$ -- a significantly larger Boolean algebra. 
Dually, this corresponds to the fact that for an infinite Boolean space, the 
Stone-{\v C}ech compactification of the underlying set, viewed as a 
discrete space, will be much larger than the original space.

Returning to our square of functors, note that 
the commutativity of the square means that we can understand
$A^\sigma$ in terms of the dual space $S(A)=(X,\tau)$. That is,
$A^\sigma={\mathcal U}(X,\leq)$ is the lattice of upsets of the dual
space of $A$ equipped with the specialisation order of the Stone
topology $\tau$. The embedding of $A$ in its canonical extension in
this description is given by the Stone embedding map $a\mapsto
\hat{a}$ which maps each element of the lattice to the corresponding
compact open upset. So canonical extension can be obtained via duality
and for this reason it is often referred to as the `double dual' in
the logic literature.

Most interestingly, the converse is also true: It is possible to
reconstruct the dual space of $A$ from the canonical extension
$A\hookrightarrow A^\sigma$ and this is why we can claim that the
theory of canonical extensions may be seen as an algebraic formulation
of Stone/Priestley duality. Given the canonical extension
$A\hookrightarrow A^\sigma$ of a DL, we obtain the dual space of $A$
by applying the discrete duality to obtain the set
$X=J^\infty(A^\sigma)$. The topology is then generated by the
`shadows' of the elements of $A$ on $X$, that is, by the sets
$\hat{a}=\{x\in X\mid x\leq a\}$ where $a$ ranges over $A$.

We point out two advantages of the canonical extension approach 
to duality. Firstly, canonical extension is particularly well-suited for
studying additional operations on lattices or Boolean algebras. This
was the original purpose for canonical extensions and their scope 
has been expanded in a modular fashion \cite{GJ2004, GH2001, DGP05} 
in order to provide representation theorems for lattice- and even 
poset-based algebras. The two-sided aspect is particularly important
when additional operations that are order-reversing are present.
Secondly, although the classical existence proof \cite{JT1951} 
for the canonical extension uses the Prime Filter Theorem, it is 
now known \cite{GH2001} that one can develop the theory of canonical 
extensions without invoking the Axiom of Choice.

\subsection{Canonical extension and logic}\label{intro:canext-and-logic}

In logic and computer science, Stone duality is central in many ways.
A landmark paper in setting this out in the clearest of terms is Abramsky's 
paper \cite{Abr1991} where he shows how Stone duality for 
distributive lattices allows us to connect specification languages with
denotational semantics. 
The role of Stone duality is similar in modal
logic in the sense that it connects specification and state-based models,
but the two approaches differ in the way they manage to factor out the 
topology inherent in Stone duality. In domain theory, one restricts to 
very special lattices and spaces for which the topology is determined 
by the specialisation order. In modal logic, one focuses on logics for which the 
topolgy `factors out' in the sense that forgetting it does not change the 
logic.

Canonical extensions are particularly
pertinent for several reasons. One is that we usually have additional
operations, like modalities, negations, or implications and the
translation of such structure as well as their equational properties
to the dual side is more easily understood by going via canonical
extension and correspondence across the discrete duality~\cite{GNV2005}.
A second and very important reason that canonical extensions play a
central role in the study of various logics is that they are centrally
related to relational semantics for these logics.

We illustrate this with the example of classical propositional modal logic, 
and we will give a very brief impression of the role
that canonical extensions play in the model theory of modal logic as it is
described in \cite{BdRV2001}.  We will consider the following two
natural semantics for modal logic:
\begin{itemize}
\item Kripke frames, which are set-based transition systems or
  coalgebras for the covariant powerset functor to be more precise,
\item modal algebras: Boolean algebras with an additional unary
(finite join preserving) operation, meant to interpret the modal diamond operator.
\end{itemize}
The former provide the natural semantics for modal logic and are 
central in various state-based models in computer science. The 
latter provide a specification language for these systems and 
often correspond to the syntactic description of the pertinent 
logics.

Thus, for classical modal logics, the restriction of the above square 
to Boolean algebras is the appropriate one, and then the additional 
structure is superposed: a modal operator on the Boolean algebras 
translates to a binary relation with certain topological properties on 
the dual spaces - this is what is known as descriptive general frames.
Forgetting the topology yields Kripke frames, which are in a discrete
duality with complex modal algebras. Note that while the inner and
outer square still commute the vertical functors are only reflectors
for the underlying Boolean algebras: this is extended Stone duality
and not natural duality for modal algebras. 
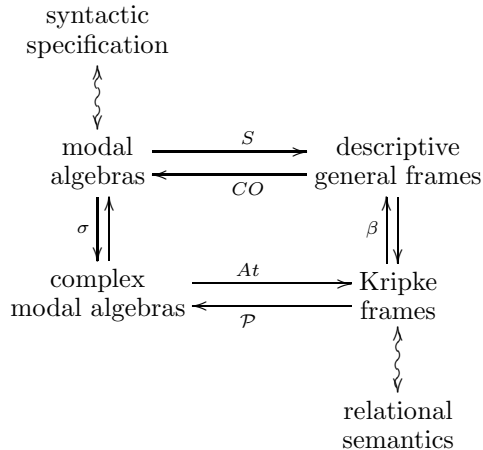
\[ \xymatrix@C=10ex{ 
\txt{syntactic\\specification} \ar@{<~>}[d] &\\
\txt{modal\\algebras} \ar[d]_{\sigma}
  \ar@<1ex>[r]^{S}&
  \txt{descriptive\\general frames} \ar@<1ex>[l]^{CO} \ar[d]\\
  \txt{complex\\modal algebras} \ar@<-1ex>[u] \ar@<1ex>[r]^{At}&
  \txt{Kripke\\frames} \ar@<1ex>[u]^{\beta}\ar@<1ex>[l]^{\mathcal{P}} \\
&\txt{relational\\semantics} \ar@{<~>}[u]
 }\]

 The central importance of canonical extension in this setting comes from 
 the fact, mentioned above, that the two important spots in the above diagram
 are the upper left and the lower right: the upper left corresponds to the 
 syntactic specification of the logic; the lower right to the semantic 
 specification. Thus moving horizontally is not enough; we must also 
 move up and down. 
In addition, we claim that the route down-and-over may be viewed as 
separating the issues involved better than the route over-and-down. To 
this end, one can think of the upper left-hand corner as the finitary 
description of the base of a topological space, and of the lower right-hand 
corner as the points underlying the space. Taking the canonical extension, 
i.e.~going down from the upper left hand corner, corresponds to augmenting 
the finitary description of the base with infinitary (but point-free) information; 
subsequently going over adds points to the picture. If we go over and 
down, already the first step (of going over) simultaneously moves us to 
a topological and point-based perspective, while going down just forgets
part of what we have worked hard to identify in the topological duality. 
Note that this separation of topological and contravariant content of the 
topological duality is even useful if our final goal is full-fledged topological 
duality (i.e., the upper right-hand corner) and not just the lower right-hand 
corner where the topology has been removed since, as we outlined in the 
previous subsection, the canonical extension, $A\hookrightarrow A^\sigma$ 
(but not $A^\sigma$ alone) contains all the topological information of the 
topological duality in a point-free and co-variant way.

 Finally, consider 
 the question of logical completeness. Given the way Kripke 
 semantics is defined, a formula $\phi$ is valid in a structure if and 
 only if the identity $\phi\approx 1$ holds in the corresponding 
 complex algebra. This is essentially the definition. On the other 
 hand, a syntactic specification of a modal logic is typically an 
 equational theory, $\Sigma$, of modal algebras. Thus soundness
 with respect to a class $\mathcal K$ of structures means that the 
 complex algebras of the structures in $\mathcal K$ all are models of 
 $\Sigma$. Completeness, in the contrapositive, means that an equation 
 that is not a consequence of $\Sigma$ is violated in the complex algebra of 
 some $K\in \mathcal K$. Canonicity of $\Sigma$ means that the class of 
 models of $\Sigma$ is closed under canonical extension.
Any equation that isn't a consequence of a theory $\Sigma$ is violated by 
 some abstract algebra model of $\Sigma$ and thus also by its canonical 
 extension.  If $\Sigma$ is canonical then this canonical extension is a 
 model of the theory in which the given equation is violated. In this way
 canonicity implies that the logic possesses complete Kripke semantics.
 One should note that not all modal logics are canonical but most of the 
 standard ones are. However, even in the absence of canonicity, it is clear 
 that canonical extensions are pertinent since they provide an account of
 the connection between the upper left and lower right corner of the
 diagram.  

\section{Preliminaries} \label{section:preliminaries}

We collect here the main facts on dcpo completions, free co-directed
completions, and canonical extensions that we will need and give
specific references to where one can find proofs.

\subsection{DCPO and suplattice presentations}

The following facts about dcpo presentations,
suplattice presentations, and dcpo algebras may be
found in \cite{JMV2008}.

\begin{definition}\label{dcpo:Definition-Presentation}
  A \emph{dcpo presentation} is a triple $\langle P; \sqsubseteq,
  C\rangle$ where
  \begin{itemize}
  \item $\langle P, \sqsubseteq\rangle$ is a preorder,
  \item $C\subseteq P \times \mathcal{P}(P)$ is a family of
    \emph{covers}, where $U$ is directed for every $(x,U)\in C$. We
    write $x \triangleleft U$ if $(x, U)\in C$. 
  \end{itemize} 
  Let $\langle D,\leq\rangle$ be a dcpo and let $f\colon P \to D$ be an
order-preserving map. We say $f$ \emph{preserves covers} if for all
  $x \triangleleft U$ it is true that $f(x) \leq \bigvee_{y\in U} f(y)$. Note that, 
  from here on, we will refer to maps preserving either an order or a preorder 
  as order-preserving in order to lighten the notation.
\end{definition}
A \emph{suplattice} is a complete join-semilattice; the appropriate
homomorphisms between suplattices are those maps which preserve all
joins. If we replace `dcpo' by `suplattice' in Definition
\ref{dcpo:Definition-Presentation} and if we drop the assumption that
each $U$ above is directed, we obtain the definition for a
\emph{suplattice presentation}.  Observe that every dcpo presentation
is also a suplattice presentation.
\begin{definition}
  A dcpo $\overline P$ is \emph{freely generated} by the dcpo
  presentation $\langle P; \sqsubseteq, C\rangle$ if there is a map
  $\eta \colon P \to \overline P$ that preserves covers, and for every
  dcpo $\langle D,\leq\rangle$ and cover-preserving map $f\colon P \to D$ there
  is a unique Scott-continuous map $\overline f \colon \overline P \to
  D$ such that $\overline f \circ \eta = f$.
\end{definition}
Again, if we replace `dcpo' with `suplattice' and `Scott-continuous
map' by `suplattice homomorphism' above, we obtain the definition of a
suplattice freely generated by a suplattice presentation. 
We will now describe how freely generated dcpos and suplattices are 
obtained in \cite{JMV2008}.
\begin{definition}
  A \emph{$C$-ideal of $P$} is a set $X \subseteq P$ which
  is downward closed and closed under covers, i.e.~for all $x
  \triangleleft U$, if $U\subseteq X$ then $x\in X$.  We denote the
  set of all $C$-ideals of $P$ by $C\operatorname{-Idl}(P)$.
\end{definition}
An arbitrary intersection of $C$-ideals is again a $C$-ideal; thus the
collection of all $C$-ideals of $\langle P; \sqsubseteq, C\rangle$
forms a complete lattice $C\operatorname{-Idl}(P)$  and we can denote
by $\langle X\rangle$ the smallest $C$-ideal containing $X$ for any
$X\subseteq P$; we will
abbreviate $\langle \{ x\} \rangle$ as $\langle x\rangle$.  Observe
that $\downarr X \subseteq \langle X \rangle$. We will denote
meets and joins in $C\operatorname{-Idl}(P)$ by $\bigwedge$ and
$\bigvee$, respectively. Note that for all $S\subseteq
C\operatorname{-Idl}(P)$, $\bigwedge S = \bigcap S$ and $\bigvee S =
\langle \bigcup S \rangle$.

\begin{proposition}[\cite{JMV2008}, Proposition 2.5] \label{suplattice:Presentation}
  Let $\langle P; \sqsubseteq, C\rangle$ be a suplattice
  presentation. Then $\langle C\operatorname{-Idl}(P),
  \subseteq\rangle$ is the suplattice freely generated by $\langle P;
  \sqsubseteq, C\rangle$, where $\eta \colon P \to \overline P$ is
  defined by $\eta \colon x \mapsto \langle x\rangle$.
\end{proposition}
\begin{definition}\label{dcpo:Overline-P}
  Given a dcpo presentation $\langle P; \sqsubseteq, C\rangle$, we
  define
  \begin{align*} \overline P = \bigcap \{ X\subseteq
    C\operatorname{-Idl}(P) \mid & X \text{ is closed under directed
      joins and }\\& \langle x \rangle \in X \text{ for all } x\in
    P\}.
  \end{align*}
\end{definition}
\begin{proposition}[\cite{JMV2008}, Theorem 2.7]
  Let $\langle P; \sqsubseteq, C\rangle$ be a dcpo presentation. Then
  $\langle \overline P, \subseteq \rangle$ is the dcpo freely
  generated by $\langle P; \sqsubseteq, C\rangle$, where $\eta \colon
  P \to \overline P$ is defined by $\eta \colon x \mapsto \langle
  x\rangle$.
\end{proposition}
Observe that it is `hard' to tell which $C$-ideals belong to
$\overline P$; see the comments at the end of Section 2 of \cite{JMV2008}.

\subsection{DCPO algebras}
We now turn to algebras. A \emph{pre-ordered
  algebra} for a set of operation symbols $\Omega$ with
arities $\alpha \colon \Omega \to \mathbb{N}$ consists of a pre-order $\langle P,
\sqsubseteq\rangle$ and order-preserving 
maps $\omega_P\colon P^{\alpha(\omega)} \to P$ for $\omega\in\Omega$. 
For dcpo presentations $\langle P_1; \sqsubseteq, C_1\rangle, \dotsc, \langle P_n;
\sqsubseteq, C_n\rangle, \langle P';\sqsubseteq, C'\rangle$ we write $x_i\triangleleft_i U_i$ if $(x_i,U_i) \in C_i$. An order-preserving map $f\colon P_1\times \dotsb
\times P_n \to P'$ is called \emph{cover-stable} if for all $1\leq i \leq n$,
all $(x_1, \dotsc, x_n)\in P_1 \times \dotsb \times P_n$ and all $U_i
\subseteq P_i$ such that $x_i
\triangleleft U_i$, we have \[ f(x_1,\dotsc, x_n) \triangleleft' \{
f(x_1,\dotsc, x_{i-1}, y, x_{i+1}, x_n) \mid y\in U_i\}.\]
\begin{proposition}[\cite{JMV2008}, Theorem 3.6]\label{dcpo:coverstable}
  If $f\colon P_1\times \dotsb \times P_n \to P'$ is cover-stable and
  order-preserving, then the function $\overline{f}\colon
  \overline{P_1}\times \dotsb \times \overline{P_n} \to
  \overline{P'}$, defined by
  \[\overline{f} 
 \colon (X_1,\dotsc, X_n) \mapsto \langle \{ f(x_1,\dotsc,x_n) \mid (x_1,
  \dotsc, x_n) \in X_1 \times \dotsb \times X_n \} \rangle, \] is
 a well-defined and Scott-continuous extension of $f$ (and is
  unique as such). 
\end{proposition}

\begin{proposition}[\cite{JMV2008}, Proposition 4.2]\label{dcpo:Canonicity}
  Consider a structure $\langle P; \sqsubseteq, C,
  (\omega_P)_{\omega\in \Omega}\rangle$ such that $\langle P;
  \sqsubseteq, C\rangle$ is dcpo presentation and $\langle P;
  \sqsubseteq,(\omega_P)_{\omega\in \Omega}\rangle$ is a
  preordered algebra.  Let $s(x_1,\dotsc,x_n)$ and $t(x_1,\dotsc ,
  x_n)$ be $n$-ary $\Omega$-terms.  If for every $\omega \in \Omega$,
  $\omega_P\colon P^{\alpha(\omega)} \to P$ is cover-stable,
  then we can define an $\Omega$-algebra structure on $\overline P$ by
  taking $\omega_{\overline P} := \overline{ \omega_P}$ and
  $P \models s\preccurlyeq t$ implies $\overline P \models
  s\preccurlyeq t$.
\end{proposition}

\subsection{Free directed completions}

The free directed join completion and the free co-directed meet completion of a 
poset are given by the posets of filters and of ideals of the poset, respectively.
For our purposes, an abstract characterisation of these completions will be important.
The following results 
date back to \cite{Schmidt2} and 
are very well known. Sources for this material are \cite{Plotkin81}, Section~6, and
\cite{GHK:03}, Sections I-4 and IV-1 and \cite{GP2008}.

\begin{definition}\label{Def:co-Completion}
  Let $\mathbb{P} = \langle P, \leq\rangle$ be a poset. By
  $\uparr_\mathbb{P} \colon \mathbb{P}\to \mathcal{F}(\mathbb{P})$ we denote the
  \emph{co-directed meet completion} of $\mathbb{P}$, which is
  characterized by the following properties:
  \begin{enumerate}
  \item $\langle \mathcal{F} (\mathbb{P}), \leq \rangle$ is a co-dcpo,
  \item $\uparr_\mathbb{P} \colon \mathbb{P}\to
    \mathcal{F}(\mathbb{P})$ is an order-embedding,
  \item for every $x\in \mathcal{F} (\mathbb{P})$, 
$\{a\in P \mid x\leq \uparr_\mathbb{P} a\}$
   is co-directed and $x= \bigwedge \{ \uparr_\mathbb{P} a \mid x\leq \uparr_\mathbb{P} a\}$,
  \item for all co-directed $S \subseteq \mathcal{F} (\mathbb{P})$ and
    all $a\in P$, if $\bigwedge S 
   \leq \uparr_\mathbb{P} a$ then there exists
    $s\in S$ such that $s
   \leq \uparr_\mathbb{P} a$.
  \end{enumerate}
\end{definition}
\begin{proposition}
  If $\mathbb{P}$ and $\mathbb{Q}$ are posets, then
  $\mathcal{F} (\mathbb{P}\times \mathbb{Q}) \cong
  \mathcal{F} (\mathbb{P}) \times \mathcal{F} (\mathbb{Q})$.
\end{proposition} 
If $f\colon \mathbb{P} \to \mathbb{Q}$ is an order-preserving map
between posets, 
then $f$ has a unique co-Scott continuous extension, 
$f^\mathcal{F} \colon \mathcal{F}(\mathbb{P}) \to\mathcal{F}(\mathbb{Q})$, 
defined as follows: \[ f^\mathcal{F} \colon x
\mapsto \bigwedge \{\uparr_\mathbb{Q} f(a) \mid x \leq
\uparr_\mathbb{P} a\}.\] 
Given an ordered algebra $\mathbb{A}=\langle A, \leq;
(\omega_\mathbb{A})_{\omega\in \Omega} \rangle$ such that every
$\omega_\mathbb{A}$ is order-preserving, we can define an
algebra structure on $\mathcal{F}(\mathbb{A})$ by taking
$\omega_{\mathcal{F}(\mathbb{A})} :=
(\omega_\mathbb{A})^\mathcal{F}$. 
\begin{proposition} \label{Filter:Canonicity}
  Let $s(x_1,\dotsc,x_n)$ and $t(x_1,\dotsc , x_n)$ be $n$-ary
  $\Omega$-terms and let $\mathbb{A}$ be an ordered
  $\Omega$-algebra. If $\mathbb{A}\models s\preccurlyeq t$ then also
  $\mathcal{F}(\mathbb{A}) \models s\preccurlyeq t$.
\end{proposition}

\begin{proposition} \label{Filters:Lattice} Let $\mathbb{A} = \langle
  A; \wedge,\vee,0,1\rangle$ be a lattice.  Then $\langle
  \mathcal{F}(\mathbb{A}), \wedge^\mathcal{F}, \vee^\mathcal{F}, 0,1
  \rangle$ is a (complete) lattice and $\uparr_\mathbb{A} \colon
  \mathbb{A} \to \mathcal{F}(\mathbb{A})$ is a lattice embedding.
\end{proposition}
We denote the meet and join operation of $\mathcal{F}(\mathbb{A})$ by
$\wedge$ and $\vee$ respectively; also, we will let $\bigwedge$
denote arbitrary meets in $\mathcal{F}(\mathbb{A})$. 
Given lattices $\mathbb{A}_1, \dotsc, \mathbb{A}_n, \mathbb{B}$, we say $f\colon
A_1 \times \dotsb \times A_n \to B$ is an
\emph{operator} if for every $1 \leq i \leq n$, all $a_i, b_i \in
A_i$ and all $a_j\in A_j$, $j \neq i$, we have
\begin{gather*}
  f(a_1, \dotsc, a_{i-1}, a_i \vee b_i, a_{i+1}, \dotsc, a_n) =\\
  f(a_1, \dotsc, a_{i-1}, a_i, a_{i+1}, \dotsc, a_n) \vee f(a_1,
  \dotsc, a_{i-1}, b_i, a_{i+1}, \dotsc, a_n).  
\end{gather*}
\begin{proposition}\label{Filter:Operator}
  If $f\colon \mathbb{A}_1 \times \dotsb \times \mathbb{A}_n \to
  \mathbb{B}$ is an operator, then so is $f^\mathcal{F} \colon
  \mathcal{F}(\mathbb{A}_1) \times \dotsb \times
  \mathcal{F}(\mathbb{A}_n) \to \mathcal{F}(\mathbb{B})$.
\end{proposition}

\subsection{Canonical extension}
Below we introduce the canonical extension  of a
lattice and the canonical extension of an order-preserving map
between lattices \cite{ GH2001}. 
Let $\mathbb{A}$ be a lattice. A \emph{lattice completion} of
$\mathbb{A}$ is a lattice embedding $e\colon \mathbb{A} \to
\mathbb{C}$ of $\mathbb{A}$ into a complete lattice
$\mathbb{C}$. Two completions of $\mathbb{A}$, $e_1\colon
\mathbb{A} \to \mathbb{C}_1$ and $e_2\colon \mathbb{A} \to
\mathbb{C}_2$, are isomorphic if there exists a lattice isomorphism
$f\colon \mathbb{C}_1 \to \mathbb{C}_2$ such that $fe_1=e_2$.
\begin{definition}
  Let $e\colon \mathbb{A} \to \mathbb{C}$ be a
  lattice completion of $\mathbb{A}$. We call $e\colon \mathbb{A} \to
  \mathbb{C}$ a \emph{canonical extension} of $\mathbb{A}$ if the
  following two conditions hold:
  \begin{itemize}
  \item (density) for all $u,v \in C$ such that $u\nleq v$, there
    exist a filter $F\subseteq A$ and an ideal $I \subseteq A$ such
    that
    \[{\textstyle\bigwedge} e[F] \leq u,\, {\textstyle\bigwedge} e[F]
    \nleq v,\, v \leq {\textstyle\bigvee} e[I] \text{ and } u\nleq
    {\textstyle\bigvee} e[I];\]
      \item (compactness) for all ideals $I\subseteq A$ and all filters
    $F\subseteq A$, if $\bigwedge e[F] \leq \bigvee e[I]$ then there
    exist $b\in F$ and $a\in I$ such that $b\leq a$.
  \end{itemize}
\end{definition}

\begin{proposition}[\cite{GH2001}, Propositions 2.6 and 2.7]
  Every lattice $\mathbb{A}$ has a canonical extension, denoted
  $e_\mathbb{A}\colon \mathbb{A} \to \mathbb{A}^\sigma$. Moreover,
  $e_\mathbb{A}\colon \mathbb{A} \to \mathbb{A}^\sigma$ is unique up
  to isomorphism of completions.
\end{proposition}
We will omit the subscript on $e_\mathbb{A}$ if it is clear from the
context what $\mathbb{A}$ is. Given $e\colon \mathbb{A} \to
\mathbb{A}^\sigma$, we define $ K(\mathbb{A}^\sigma) := \{ {\textstyle
  \bigwedge} e[F] \mid F\subseteq A \text{ a filter}\}$ to be the
\emph{closed elements} of $\mathbb{A}^\sigma$.
\begin{definition}
  Let 
 $f\colon \mathbb{A}_1\times \dotsb \times \mathbb{A}_n\to
  \mathbb{B}$ be an order-preserving map 
 where $\mathbb{A}_1,\dotsc \mathbb{A}_n$ and $\mathbb{B}$
  are lattices. We define $f^\sigma
  \colon \mathbb{A}_1^\sigma \times \dotsb \times \mathbb{A}_n^\sigma
  \to \mathbb{B}^\sigma$ by first putting
  \[ f^\sigma \colon (x_1,\dotsc, x_n) \mapsto \bigwedge \{
  e_\mathbb{B}(f(a_1,\dotsc, a_n)) \mid (x_1, \dotsc, x_n) \leq (a_1,
  \dotsc, a_n)\}\] for all tuples of closed elements $(x_1,\dotsc,
  x_n) \in K(\mathbb{A}_1^\sigma) \times \dotsb \times
  K(\mathbb{A}_n^\sigma)$.  We then define $f^\sigma$ as follows on
  arbitrary tuples $(u_1, \dotsc, u_n) \in \mathbb{A}_1^\sigma \times
  \dotsb \mathbb{A}_n^\sigma$:
  \begin{gather*} f^\sigma \colon (u_1, \dotsc, u_n) \mapsto \bigvee
    \big\{ f^\sigma (x_1,\dotsc, x_n) \mid \\ (u_1, \dotsc, u_n) \geq
    (x_1,\dotsc, x_n) \in K(\mathbb{A}_1^\sigma) \times \dotsb \times
    K(\mathbb{A}_n^\sigma) \big\}.
  \end{gather*}
  
  For information on the naturality of this definition in the distributive setting,
  see \cite{GJ2004}, Theorem~2.15.
  \end{definition}

\section{A dcpo presentation of the canonical extension} \label{section:results}

\begin{definition} \label{Delta:Definition} Given a lattice
  $\mathbb{A}$, we define a dcpo presentation 
  \[\Delta(\mathbb{A}) :=
  \langle \mathcal{F}(\mathbb{A}); \leq, C_\mathbb{A}\rangle
  \]
 where
  \begin{align*} C_\mathbb{A}:= \big\{ (x,U) &\in
    \mathcal{F}(\mathbb{A}) \times
    \mathcal{P}(\mathcal{F}(\mathbb{A})) \mid U \text{ non-empty,
      directed},\\& \forall I \in \operatorname{Idl} (\mathbb{A})
    [(\forall x' \in U\, \exists a' \in I,\, x' \leq \uparr_\mathbb{A}
    a' )
    \Rightarrow \exists a\in I,\, x \leq \uparr_\mathbb{A} a]\big\}.
  \end{align*}
\end{definition}

We now present several properties of dcpo presentations of the shape
$\Delta (\mathbb{A})$.  Let $\eta \colon \mathcal{F}(\mathbb{A}) \to
\overline{\Delta(\mathbb{A})}$ be the natural map $x \mapsto \langle
x\rangle$.

\begin{lemma} \label{Delta:Sup-Presentation} Let $\mathbb{A}$ be a
    lattice. Then $\overline{\Delta(\mathbb{A})} =
    C_\mathbb{A}\operatorname{-Idl} (\Delta(\mathbb{A}))$ and $\eta \colon
    \mathcal{F}(\mathbb{A}) \to \overline{\Delta(\mathbb{A})}$ is a
    $\vee$-homomorphism. Consequently, every $u\in
    \overline{\Delta(\mathbb{A})}$ is a lattice ideal of
    $\mathcal{F}(\mathbb{A})$.
\end{lemma}
\begin{proof}
  We will write $\Delta, \mathcal{F}, C$, assuming $\mathbb{A}$ is
  fixed.

  We show the following stability property of $C$: for all
  $y\in\mathcal{F}$ and all $x \triangleleft U$, we have $x\vee y
  \triangleleft U \vee y$ where $U \vee y=\{x'\vee y\mid x'\in
  U\}$. To this end, suppose that $I\in \operatorname{Idl}
  (\mathbb{A})$ such that for all $x' \in U$
 there exists $a'\in I$
  such that $x' \vee y \sqsubseteq \uparr_\mathbb{A} a'$. Since $U$ is non-empty, this
  condition is non-vacuous so that $y \sqsubseteq x' \vee y
  \sqsubseteq \uparr_\mathbb{A} a'$ for some $x' \in U$ and $a' \in I$.  Moreover, since
  $x \triangleleft U$ and $x'\sqsubseteq x' \vee y$ for all $x'\in U$,
  there exists $a\in I$ such that $x \sqsubseteq \uparr_\mathbb{A} a$. But then also
  $x\vee y \sqsubseteq \uparr_\mathbb{A} a \vee \uparr_\mathbb{A} a'=
  \uparr_\mathbb{A}( a \vee a')$ where $a\vee a' \in I$, so that $x\vee y\triangleleft
  U \vee y$.  It now follows by \cite[Proposition 6.2]{JMV2008} that
  $\overline \Delta$ is the suplattice presented by $\Delta$ and that
  $\eta \colon\mathcal{F} \to \overline \Delta$ is a
  $\vee$-homomorphism. It follows by Proposition
  \ref{suplattice:Presentation} that $\overline{\Delta} =
  C\operatorname{-Idl}(\Delta)$.

  Let $u\in \overline{\Delta}$; we will show that $u$ is a
  lattice ideal of $\mathcal{F}$. It follows from
  Definition \ref{dcpo:Definition-Presentation} that $u$ is a
  down-set. Moreover, if $x,y \in u$, then $\eta(x), \eta( y)
  \subseteq u$, so that $\eta(x) \vee \eta(y) \subseteq u$. Since
  $\eta$ is a $\vee$-homomorphism, $\eta(x \vee y) \subseteq u$,
  whence $x\vee y \in u$. It follows that $u$ is a lattice ideal.
\end{proof}

\begin{remark}\label{rem:suplatvsdcpo}
  We would like to highlight that Lemma~\ref{Delta:Sup-Presentation}
  above is a crucial step in allowing the lifting of operators.  The
  canonical extension of a lattice is not just a dcpo completion but a
  suplattice completion of the free dual dcpo completion of the
  lattice. However, there is no equivalent of
  Proposition~\ref{dcpo:Canonicity} for suplattice algebras (see
  \cite[Sec. 4]{JMV2008}). The lemma tells us that
  $\overline{\Delta(\mathbb{A})}$ is in fact also the suplattice
  presented by $\Delta(\mathbb{A})$ as its elements are \emph{all}
  $C_\mathbb{A}$-ideals of $\Delta(\mathbb{A})$. The description of
  this suplattice completion as a dcpo completion is crucial as it
  implies that Proposition~\ref{dcpo:Canonicity} applies. Thus
  Lemma~\ref{Delta:Sup-Presentation} tells us that we \emph{can} lift
  inequations to suplattices with presentations of the shape
  $\Delta(\mathbb{A})$ since they are also dcpo presentations.
\end{remark}

The following Lemma will allow us to show that $\overline{\Delta(\mathbb{A})}$
is in fact the canonical extension of $\mathbb{A}$.
  
\begin{lemma} \label{Delta:Properties}
Let $\eta \colon \mathcal{F}(\mathbb{A}) \to \overline{\Delta(\mathbb{A})}$ be the natural
  map $x \mapsto \langle x\rangle$.
  \begin{enumerate}
  \item For all $x\in \mathcal{F}(\mathbb{A})$, $\eta(x) = \downarr_{\mathcal{F}(\mathbb{A})} x$, hence
    $\eta \colon \mathcal{F}(\mathbb{A}) \to \overline{ \Delta(\mathbb{A})}$ is an embedding.
  \item $\overline{ \Delta(\mathbb{A})}$ is a complete lattice.
  \item $\eta \colon \mathcal{F}(\mathbb{A}) \to \overline{ \Delta(\mathbb{A})}$ is 
  a $\vee,\bigwedge$-homomorphism.
    \item For all directed $T \subseteq A$, 
$\bigvee_{b\in T}
    \langle \uparr_\mathbb{A} b\rangle = \bigcup_{b \in T} \langle \uparr_\mathbb{A} b\rangle$.
 \end{enumerate}
\end{lemma}

\begin{proof}
  We will write $\Delta, \mathcal{F}, C$, assuming $\mathbb{A}$ is fixed.

  (1) 
  We will show that $\downarr_{\mathcal{F}(\mathbb{A})} x$ is a $C$-ideal, which is sufficient since 
  necessarily $\downarr_{\mathcal{F}(\mathbb{A})} x\subseteq \langle x\rangle$. Suppose that 
  $y \triangleleft U$ and $U
  \subseteq\downarr_{\mathcal{F}(\mathbb{A})} x$. If $a \in A$ such
  that $x \leq \uparr_\mathbb{A} a $
  then $\downarr_{\mathbb A} a$ is an ideal of $\mathbb A$ 
  and for each $x'\in U$, $x' 
 \leq x
 \leq \uparr_\mathbb{A} a$, 
  so by the definition of $C$, there is $a' \in \downarr_\mathbb{A} a$ with $y
 \leq \uparr_\mathbb{A} a'$. 
  That is, $x \leq \uparr_\mathbb{A} a$ implies $y
 \leq \uparr_\mathbb{A} a$ and thus
  \[ 
  y
 \leq \bigwedge \{ \uparr_\mathbb{A} a \mid x \leq \uparr_\mathbb{A} a\} = x
  \] 
  and  $\downarr_{\mathcal{F}(\mathbb{A})} x$ is a $C$-ideal.
  
  (2) It follows from Lemma \ref{Delta:Sup-Presentation} that
  $\overline{\Delta}$ is complete lattice.

  (3) It follows from Lemma \ref{Delta:Sup-Presentation} that $\eta$ is a
  $\vee$-homomorphism.  Let $S \subseteq \mathcal{F}$; we will show that $\bigwedge_{x\in S} \langle
  x\rangle = \langle \bigwedge S \rangle$. This follows immediately
  from the fact that $C\operatorname{-Idl}
  (\Delta)$ is a closure system and (1) above:
  \[ \bigwedge_{x \in S} \langle x\rangle= \bigcap_{x \in S} \langle x\rangle = \bigcap_{x \in S}
  \downarr_{\mathcal{F}(\mathbb{A})} x = \downarr_{\mathcal{F}(\mathbb{A})} ({\textstyle\bigwedge} S) = \langle
  {\textstyle\bigwedge} S \rangle.\] 

  (4) Since $\bigcup_{b \in T} \langle \uparr_\mathbb{A} b\rangle \subseteq \big\langle
  \bigcup_{b \in T} \langle \uparr_\mathbb{A} b\rangle \big\rangle = \bigvee_{b\in T}
  \langle \uparr_\mathbb{A} b\rangle$, it suffices to show that $\bigcup_{b \in T}
  \langle \uparr_\mathbb{A} b\rangle$ is a $C$-ideal. Let $I:= \downarr_\mathbb{A} T$.  Now
  suppose that $x \triangleleft U$ and $U \subseteq \bigcup_{b \in T}
  \langle \uparr_\mathbb{A} b\rangle = \bigcup_{b \in T} \downarr_{\mathcal{F}(\mathbb{A})} (\uparr_\mathbb{A} b)$; then for each
  $x' \in U$, there is a $b'\in I$ such that $x' 
 \leq \uparr_\mathbb{A} b'$. Since $x
  \triangleleft U$, it follows that there is some $b\in I$ such that
  $x
 \leq \uparr_\mathbb{A} b$; since $I=\downarr_\mathbb{A} T$, we may assume that $b\in
  T$. But then $x \in \bigcup_{b \in T} \langle \uparr_\mathbb{A} b\rangle$; it follows
  that $\bigcup_{b \in T} \langle \uparr_\mathbb{A} b\rangle$ is a $C$-ideal. 
    \end{proof}

\begin{remark}
  Analogous to the $\bigwedge,\vee$-homomorphism $\eta\colon
  \mathcal{F}(\mathbb{A}) \to \overline{\Delta(\mathbb{A})}$ we could
  also define a $\bigvee, \wedge$-homomorphism $\mu \colon
  \mathcal{I}(\mathbb{A}) \to \overline{\Delta(\mathbb{A})}$, where
  $\mathcal{I}(\mathbb{A})$ is the directed join-completion (or the
  ideal completion) of $\mathbb{A}$. We would then use the map $\mu
  \colon y \mapsto \bigvee_{b\in y}\langle b\rangle$.
\end{remark}

Let $e\colon \mathbb{A} \to \overline{ \Delta(\mathbb{A})}$ be the
restriction of $\eta\colon \mathcal{F}(\mathbb{A}) \to \overline{
  \Delta(\mathbb{A})}$ to $A$, i.e.\[e\colon a \mapsto \langle
\uparr_\mathbb{A} a \rangle = \downarr_{\mathcal{F}(\mathbb{A})}
\left( \uparr_\mathbb{A} a \right).\]

\begin{theorem}\label{Delta:Can-Ext}
  Let $\mathbb{A}$ be a lattice. Then the embedding $e\colon
  \mathbb{A} \to \overline{ \Delta(\mathbb{A})}$ is the canonical
  extension of $\mathbb{A}$.
\end{theorem}

\begin{proof}
  We will write $\Delta, \mathcal{F}, C$ as before. First, observe
  that it follows from Proposition \ref{Filters:Lattice} and Lemma
  \ref{Delta:Properties}.1 that $e\colon \mathbb{A} \to \overline
  \Delta$ is an embedding.

  Next, in order to prove that the embedding is dense, assume that
  $u,v \in \overline \Delta$ such that $u \nsubseteq v$. We will show
  that there are a filter $F$ and an ideal $I$ of $\mathbb{A}$ such
  that $\bigwedge e[F] \subseteq u$, $\bigwedge e[F] \nsubseteq v$, $u
  \nsubseteq\bigvee e[I]$ and $v\subseteq \bigvee e[I]$. It follows
  from $u\nsubseteq v$ that there is some $x\in u \setminus v$, so
  that $\langle x\rangle \subseteq u$ and $\langle x\rangle \nsubseteq
  v$. Take $F:= \{a \in A \mid x \leq \uparr_\mathbb{A} a\}$, then $\langle x\rangle =
  \bigwedge e[F]$ and we have our first witness; we will use this same
  element $x\in u\setminus v$ to find a suitable ideal $I$. Now
  observe that $v$ is a directed subset of $\mathcal{F}$ by Lemma
  \ref{Delta:Sup-Presentation}. If it were the case that
  $x\triangleleft v$, then since $v$ is a $C$-ideal and $v\subseteq
  v$, it would follow that $x\in v$, contrary to our assumption. So it
  must be the case that $x \ntriangleleft v$ and thus, by the
  definition of the covering relation, there must be some ideal
  $I\subseteq A$ such that
  \begin{equation} \label{IdealWitness} \forall \,x'\in v, \exists\,
    a'\in I \text{ such that }x' 
    \leq \uparr_\mathbb{A} a', \text{ but }
    \forall\, a\in I, \, x\not
   \leq \uparr_\mathbb{A} a.
  \end{equation}
  We claim that $u \nsubseteq\bigvee e[I]$ and $v\subseteq \bigvee
  e[I]$. If the former were the case, then we would find that \[x\in u
  \subseteq \bigvee e[I] = \bigvee_{a\in I} \langle \uparr_\mathbb{A} a\rangle =
  \bigcup_{a\in I} \langle \uparr_\mathbb{A} a\rangle = \bigcup_{a\in I} \downarr_{\mathcal{F}(\mathbb{A})}
  (\uparr_\mathbb{A} a),\] where the last two equalities follow from
  Lemma~\ref{Delta:Properties}. It now follows that $x
 \leq \uparr_\mathbb{A} a$ for
  some $a\in I$, contradicting \eqref{IdealWitness}. Finally, given
  $x' \in v$ and $a' \in I$ such that $x'
 \leq
  \uparr_\mathbb{A} a'$, we find that $\langle x' \rangle \subseteq \langle \uparr_\mathbb{A} a'\rangle$, so that
  it follows from \eqref{IdealWitness} that
  \[ v= \bigvee \{ \langle x'\rangle \mid x'\in v\} \subseteq \bigvee \{
  \langle \uparr_\mathbb{A} a' \rangle \mid a' \in I\} = \bigvee e[I].\]

  Finally, for the compactness property, suppose that $F$ and $I$ are an 
  arbitrary filter and ideal of $\mathbb{A}$ such that 
  $\bigwedge e[F] \subseteq \bigvee e[I]$;
  we must show that there exists $a\in I$ and $b\in F$ such that $b
 \leq a$. By Lemma \ref{Delta:Properties}.3, $\bigwedge
  e[F]=\langle \bigwedge F \rangle$, so we find that \[ \bigwedge F
  \in \langle \bigwedge F \rangle =\bigwedge e[F] \subseteq \bigvee
  e[I] = \bigcup_{a\in I} \downarr_{\mathcal{F}(\mathbb{A})} (\uparr_\mathbb{A} a), \] where the second equality
  follows from Lemma~\ref{Delta:Properties}.4 as before. It follows that
  $\bigwedge F \in \downarr_{\mathcal{F}(\mathbb{A})} (\uparr_\mathbb{A} a)$ for some $a\in I$, so by Definition
  \ref{Def:co-Completion}.4, there is some $b\in F$ such
  that $b
 \leq a $. 
\end{proof}

Recall that if $\mathbb{A}$ is a lattice and $e\colon A \to
\mathbb{A}^\sigma$ is its canonical extension, the \emph{closed
  elements} of $\mathbb{A}^\sigma$ are defined as
\[ K(\mathbb{A}^\sigma) := \{{\textstyle \bigwedge} e[F] \mid F
\subseteq A, F \text{ a filter}\}.\] If we view $\overline{\Delta(\mathbb{A})}$ as the
canonical extension of $\mathbb{A}$, then the closed elements
correspond to the elements of $\mathcal{F}(\mathbb{A})$:
\[ K(\overline{\Delta(\mathbb{A})}) = \{ \langle x\rangle \mid x\in
\mathcal{F}(\mathbb{A}) \}.\] This follows from the fact that for each
$x\in \mathcal{F}(\mathbb{A})$, $\{ a\in A \mid x \leq
\uparr_\mathbb{A} a \}$ is a filter and we have $x = \bigwedge \{
\uparr_\mathbb{A} a \mid x \leq
\uparr_\mathbb{A} a \}$, and the fact that $\eta \colon \mathcal{F}(\mathbb{A}) \to
\overline{ \Delta(\mathbb{A})}$ preserves all meets by Lemma
\ref{Delta:Properties}.3.

\begin{lemma}\label{Delta:Operator-Cover-Stable}
  Let $\mathbb{A}_1,\dotsc, \mathbb{A}_n, \mathbb{B}$ be lattices and
  let $f\colon \mathbb{A}_1 \times \dotsb \times \mathbb{A}_n \to
  \mathbb{B}$ be an operator.  Then $f^\mathcal{F} \colon
  \mathcal{F}(\mathbb{A}_1)\times \dotsb \times \mathcal{F}
  (\mathbb{A}_n) \to \mathcal{F} (\mathbb{B})$ is cover-stable.
\end{lemma}
\begin{proof}
  We write $x_i \triangleleft_i U_i$ if $(x_i, U_i) \in C_{\mathbb{A}_i}$ and
  $x\triangleleft U$ if $(x,U) \in C_\mathbb{B}$. Let $1 \leq i \leq
  n$, $(x_1, \dotsc, x_n) \in \mathcal{F}(\mathbb{A}_1) \times \dotsb
  \times \mathcal{F}(\mathbb{A}_n)$ and $U_i\subseteq
  \mathcal{F}(\mathbb{A}_i)$ such that $x_i \triangleleft_i U_i$. We
  need to show that
  \begin{equation}\label{f:Stable}
    f^\mathcal{F} ( x_1,
    \dotsc, x_n) \triangleleft \{ f^\mathcal{F} ( x_1, \dotsc, x_{i-1}, y,
    x_{i+1}, \dotsc, x_n) \mid y \in U_i\}.
  \end{equation}
  We will write $f^\mathcal{F} (-,y,-)$ for an element of the right
  hand side set above.  Let $I \in \operatorname{Idl}(\mathbb{B})$
  such that for every $y\in U_i$, there is some $a_y\in I$ such that
  $f^\mathcal{F} (-,y,-) 
  \leq \uparr_{\mathbb{B}} a_y$. We need to find some $c\in I$ such that $f^\mathcal{F}(-,x_i,-) 
  \leq \uparr_{\mathbb{B}} c$. Now since
  $f^\mathcal{F}$ is co-Scott continuous, it is also co-Scott
  continuous in its $i$th coordinate \cite[Lemma 3.2.6]{AJ1994}.  Thus if
  we take $y\in U_i$ and write $y= \bigwedge \{
\uparr_{\mathbb{A}_i} b \mid y \leq
\uparr_{\mathbb{A}_i} b \}$,
  then
  \begin{align*} 
    f^\mathcal{F} (-,y,-) &= f^\mathcal{F} (-,{\textstyle \bigwedge} \{
    \uparr_{\mathbb{A}_i} b \mid y \leq
    \uparr_{\mathbb{A}_i} b \},-) \\ & = \bigwedge_{b\in A_i, y \leq \uparr_{\mathbb{A}_i} b}
    f^\mathcal{F} (-,\uparr_{\mathbb{A}_i} b,-)  \leq
    \uparr_{\mathbb{B}} a_y.
  \end{align*}
  It follows by Definition
  \ref{Def:co-Completion}.4 that there is some $b_y\in A_i $
  such that $y \leq \uparr_{\mathbb{A}_i} b_y$ and $f^\mathcal{F} (-,y,-)  \leq f^\mathcal{F} (-,\uparr_{\mathbb{A}_i} b_y,-)\leq \uparr_{\mathbb{B}} a_y$. Let 
  $I' \in \operatorname{Idl} (\mathbb{A}_i)$ be
  the ideal generated by $\{b_y \mid y\in U_i\}$. Since 
  $y \leq \uparr_{\mathbb{A}_i} b_y \in I'$ for each $y \in U_i$ and $x_i \triangleleft_i U_i$, 
  it follows that there is some $b\in I'$ such that $x_i \leq \uparr_{\mathbb{A}_i} b$. 
  By definition of $I'$, there exist $y_1, \dotsc, y_k \in U$ such that $x_i
  \leq \uparr_{\mathbb{A}_i} b \leq \uparr_{\mathbb{A}_i} b_{y_1} \vee \dotsb \vee \uparr_{\mathbb{A}_i} b_{y_k}$. But
  then
  \begin{align*} 
  f^\mathcal{F}(-,x_i,- ) 
   & \leq
    f^\mathcal{F}(-,\uparr_{\mathbb{A}_i} b,-) \\
       & \leq f^\mathcal{F}(-,\uparr_{\mathbb{A}_i} b_{y_1} \vee
    \dotsb \vee \uparr_{\mathbb{A}_i} b_{y_k},-) \\
    & = f^\mathcal{F}(-,\uparr_{\mathbb{A}_i} b_{y_1},-) \vee
    \dotsb \vee f^\mathcal{F}(-,\uparr_{\mathbb{A}_i} b_{y_k},-) \\
      & \leq \uparr_{\mathbb{B}} a_{y_1}
    \vee \dotsb \vee \uparr_{\mathbb{B}} a_{y_k} = \uparr_{\mathbb{B}}
    (a_{y_1}
  \vee \dotsb \vee a_{y_k}), 
  \end{align*} 
  where the first equality follows from the fact that $f^\mathcal{F}$ is an
  operator (by Proposition \ref{Filter:Operator}). Since $a_{y_1}
  \vee \dotsb \vee a_{y_k} \in I$ and $I$ was
  arbitrary, it follows that \eqref{f:Stable} holds.
\end{proof}

\begin{corollary}
  Let $\mathbb{A}_1, \dotsc, \mathbb{A}_n$ and $\mathbb{B}$ be
  lattices and let $f\colon \mathbb{A}_1 \times \dotsb \times
  \mathbb{A}_n \to \mathbb{B}$ be an operator. Then
  $\overline{f^\mathcal{F}}\colon \overline{\Delta(\mathbb{A}_1)} \times \dotsb
  \times \overline{\Delta(\mathbb{A}_n)} \to \overline{\Delta(\mathbb{B})}$ 
  is well-defined and Scott-continuous. Moreover, $\overline{f^\mathcal{F}} =
  f^\sigma$.
\end{corollary}
\begin{proof}
  Let $f\colon \mathbb{A}_1 \times \dotsb \times \mathbb{A}_n \to
  \mathbb{B}$ be as in the assumptions above.  It follows from
  Proposition \ref{Filter:Operator} 
 and Lemma
  \ref{Delta:Operator-Cover-Stable} that $\overline{f^\mathcal{F}}$ is
  well-defined and Scott-continuous. To show that
  $\overline{f^\mathcal{F}} = f^\sigma$; 
  observe that
  $\overline{f^\mathcal{F}}$ and $f^\sigma$ agree on closed elements:
  \[ \overline{f^\mathcal{F}}(\langle x_1\rangle, \dotsc, \langle x_n \rangle)=
  \langle f^\mathcal{F}(x_1,\dotsc,x_n) \rangle,\] by \cite[Lemma
  3.3]{JMV2008}. 
  Since $x_i = \bigwedge \{\uparr_{\mathbb{A}_i} b \mid x_i \leq \uparr_{\mathbb{A}_i} b $ for
  all $1\leq i \leq n$, we find that
  \begin{gather*} \langle f^\mathcal{F}(x_1,\dotsc,x_n) \rangle=\\
    \langle f^\mathcal{F}({\textstyle \bigwedge} \{
    \uparr_{\mathbb{A}_1} a_1 \mid x_1 \leq \uparr_{\mathbb{A}_1}
    a_1\},\dotsc, {\textstyle \bigwedge} \{
    \uparr_{\mathbb{A}_1} a_n \mid x_n \leq \uparr_{\mathbb{A}_n} a_n\})
    \rangle= \\ \bigwedge \big\{ \langle \uparr_{\mathbb{B}} f(a_1, \dotsc, a_n) \rangle
    \mid (x_1, \dotsc, x_n) \sqsubseteq (\uparr_{\mathbb{A}_1} a_1, \dotsc, \uparr_{\mathbb{A}_n} a_n) \big\} =
    f^\sigma(\langle x_1\rangle, \dotsc, \langle
    x_n\rangle),\end{gather*} where the second equality follows from
  the fact that both $f^\mathcal{F}$ and $\langle \cdot \rangle$
  commute with co-directed meets.

  Secondly, recall from Lemma \ref{Delta:Sup-Presentation} that every
  $u \in \overline{\Delta(\mathbb{A}_i)}$, seen as a $C$-ideal, is a
  directed subset of $\mathcal{F}(\mathbb{A})$. Thus, $u=
  \bigvee_{x\in u} \langle x \rangle$ is a directed join. Since we
  showed above that $\overline{f^\mathcal{F}}$ is Scott-continuous, it
  follows that
  \begin{gather*}
    \overline{f^\mathcal{F}}(u_1, \dotsc, u_n)=
    \overline{f^\mathcal{F}} \big( {\textstyle\bigvee_{x_1\in u_1}}
    \langle x_1 \rangle, \dotsc,{\textstyle\bigvee_{x_n\in u_n}}
    \langle x_n\rangle\big) =\\ 
    \bigvee \big\{ \overline{f^\mathcal{F}} ( \langle x_1 \rangle,
    \dotsc, \langle x_n\rangle) \mid x_i \in u_i \text{ for all } 1
    \leq i \leq n \big\} =\\ 
    \bigvee \big\{ f^\sigma ( \langle x_1 \rangle, \dotsc, \langle
    x_n\rangle) \mid x_i \in u_i \text{ for all } 1 \leq i \leq n
    \big\} = f^\sigma (u_1,\dotsc, u_n),
  \end{gather*}
  for arbitrary $(u_1,\dotsc, u_n) \in \overline{\Delta(\mathbb{A}_1)}
  \times \dotsb \times \overline{\Delta(\mathbb{A}_n)}$.
\end{proof}
Thus, we have shown that the dcpo presentation $\Delta(\mathbb{A})$ of
Definition \ref{Delta:Definition} allows us to describe the canonical
extension of a lattice $\mathbb{A}$, together with the
$\sigma$-extension of any additional operator $f\colon A^n \to A$. 
The
following theorem, which can be found in \cite{GJ1994, GH2001}, can
now be seen as an application of general results concerning dcpo
algebras from \cite{JMV2008} to the specific case of canonical
extensions of lattices with operators.
\begin{theorem} [cf.~\cite{GJ1994}, Theorem 4.5 and \cite{GH2001}, Theorem 6.3]
\label{GJ94}
  Let $\mathbb{A}$\linebreak $=\langle A; \wedge_\mathbb{A}, \vee_\mathbb{A},
  0_\mathbb{A}, 1_\mathbb{A}, (\omega_\mathbb{A})_{\omega \in \Omega'}
  \rangle$ be a bounded lattice with additional operations and let 
  $\Omega\subseteq\{\wedge,\vee,0,1\} \cup \Omega'$ consist entirely 
  of operation symbols that interpret as operators in $\mathbb{A}$. If
  $s(x_1,\dotsc,x_n)$ and 
  $t(x_1,\dotsc , x_n)$ are $n$-ary $\Omega$-terms such that
  $\mathbb{A} \models s \preccurlyeq t$, then also $\mathbb{A}^\sigma \models
  s \preccurlyeq t$.
\end{theorem}
\begin{proof}
  Let $\mathbb{A}$, $s$ and $t$ be as in the assumptions of the
  theorem.  Since operators are monotone, it follows by
  Proposition \ref{Filter:Canonicity} that $\mathcal{F}(\mathbb{A})
  \models s \preccurlyeq t$. It follows by Proposition
  \ref{Filter:Operator} and Lemma \ref{Delta:Operator-Cover-Stable}
  that $\overline{\Delta(\mathbb{A})} \models s\preccurlyeq t$.
\end{proof}

\begin{remark}
  Observe that $\vee_\mathbb{A} \colon \mathbb{A} \times \mathbb{A}\to
  \mathbb{A}$ is always an operator by associativity but that
  $\wedge_\mathbb{A}\colon \mathbb{A} \times \mathbb{A} \to\mathbb{A}$
  is an operator if and only if $ \mathbb{A}$ is distributive.
\end{remark}

\begin{remark}\label{rem:2sided}
Canonical extension is a two sided construction: it does not favour
joins over meets. 
This is perhaps best illustrated by \cite{GP2008}. There it is shown
that if we consider alternating applications of directed join and meet
completion to a lattice $\mathbb{A}$, then the embeddings $\downarr_{\mathcal{F}(\mathbb{A})}
\colon \mathcal{F}(\mathbb{A}) \to
\mathcal{I}(\mathcal{F}(\mathbb{A}))$ and $\uparr_{\mathcal{I}(\mathbb{A})}
\colon \mathcal{I}(\mathbb{A}) \to
\mathcal{F}(\mathcal{I}(\mathbb{A}))$ factor through
$\mathbb{A}^\delta$ in a unique way; see Figure \ref{fig:int}.
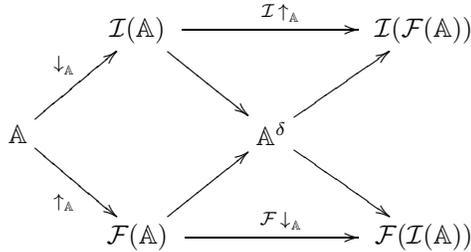
\begin{figure}[htp]    
$$ 
 \hskip-2cm 
 \xymatrix{   
&
{\mathcal{I}(\mathbb{A})\ }\ar@{->}[dr]\ar@{->}[rr]^{\mathcal{I}
\uparr_\mathbb{A}}    
&   
&{\  \mathcal{I}(\mathcal{F}(\mathbb{A}))}   
\\ 
\mathbb{A}\ar@{->}[ur]^{\downarr_\mathbb{A}} \ar@{->}[dr]_{\uparr_\mathbb{A}} &  
&   
{\mathbb{A}^\delta}\ar@{->}[ur]\ar@{->}[dr]   
&   
\\   
&
{\mathcal{F}(\mathbb{A})\ }\ar@{->}[ur]\ar@{->}[rr]^{\mathcal{F}
\downarr_\mathbb{A}}   
&   
&{\ \mathcal{F}(\mathcal{I}(\mathbb{A}))}   
 }   
$$   
 \caption{The canonical extension as an interpolant, as discussed in 
\cite{GP2008}}
\label{fig:int}     
\end{figure} 
In order to apply the existing theory on dcpo completions we have presented our
results in terms of a dcpo completion of the free co-directed meet completion of
the original lattice, using the fact that $\mathbb{A}^\delta$
interpolates between $\mathcal{F}(\mathbb{A})$ and
$\mathcal{I}(\mathcal{F}(\mathbb{A}))$. 
Of course the order dual approach would have worked just as well:
Starting from the directed join completion (concretely, the ideal completion) of 
$\mathbb{A}$, we could have given a \emph{co-dcpo}-presentation of 
$\mathbb{A}^\delta$. The extension 
of a dual operator $f\colon\mathbb{A}_1 \times \dotsb \times \mathbb{A}_n \to \mathbb{B}$,
i.e.~a map preserving binary meets in each coordinate, via this co-dcpo presentation would 
then yield an extension 
$f^\pi \colon\mathbb{A}_1^\delta \times \dotsb \times \mathbb{A}_n^\delta\to\mathbb{B}$ of
$f$ and the dual of Theorem~\ref{GJ94} would guarantee that equations among dual 
operators lift to the extension. This remark restores some symmetry to the situation, though 
we note that the extension $f^\sigma$ obtained from the free co-dcpo followed by the dcpo 
completion described in this paper and the extension of an operation obtained via the order 
dual approach do not in general agree. This latter extension is also well known and much 
used in the theory of canonical extensions and is known as the $\pi$-extension of $f$.
The extension of the underlying lattice using either approach is however one and the same -- 
this is easy to see by the fact that the characterising properties of canonical extensions 
are self-dual properties.
\end{remark}

\section{Discussion} \label{section:discussion}   
 
 The original 1951 canonicity result of J\'onsson and Tarski had a fairly complicated
 proof. In addition, it required the underlying lattice to be, not only distributive, but
 Boolean even though the canonicity of equations only is implied if the negation
 is not involved. The latter fact obviously begged the question of whether the result 
 was actually a (distributive) lattice result. 
 
 It took over 40 years before this question was answered in the positive in the 
 paper \cite{GJ1994} (and fairly soon afterwards, it was shown \cite{GH2001} 
 that it was in fact just a lattice result). The main breakthrough was in the 1994 
 paper and it consisted in realising the central role played by Scott continuity.
 Even though the paper \cite{GJ1994} was written in a language quite different 
 from that of \cite{JMV2008}, the general lines of the proof in \cite{GJ1994} do 
 in fact follow those of \cite{JMV2008}, albeit in the special case of the 
 presentation $\Delta(\mathbb{A})$. With this article we have shown explicitely
 how the two relate.
 
 While the canonicity result for operators is a special case of the much more 
 general domain theoretic result of \cite{JMV2008}, the real power and 
 interest of canonical extensions involves, at least the presence, and sometimes
 also the direct involvement of {\it order reversing} operations such as negations,
 implications, and other non-monotonic logical connectives. Because of the 
 up-down symmetry of canonical extension, order-reversing operations are
 easily and meaningfully extended to canonical extensions (we have just identified 
 it as the free dcpo generated by a dcpo presentation based on a free co-dcpo 
 completion, but as mentioned in Remark~\ref{rem:2sided}  above, we could as well 
 have obtained it as the free co-dcpo generated by a co-dcpo presentation based on 
 a free dcpo completion of the original algbera). In \cite{GJ2004} topological methods 
 for canonical extensions were introduced and these allow arbitrary maps to be extended
 to the canonical extension in a very natural way. This in turn allows for a very fine 
 analysis of canonicity in that general setting \cite{GJ2004}.  We are not aware of 
 any parallel to these methods in domain theory but expect that the current paper will 
 foster new unifying developments.
 
 As a case in point, one of the referees of this paper pointed out that our 
 Definition~\ref{Delta:Definition}, and the results following it, may be generalised 
 to a more general dcpo presentation setting. These generalisations are indeed 
 possible and this is closely related to parallel work of Sam van Gool on canonical 
 extensions of strong proximity lattices which are a kind of dcpo presentations of stably 
 compact spaces.
 
\bibliographystyle{plain}

\end{document}